\documentclass[12pt]{article}

\usepackage{latexsym,amsmath,amssymb,amsthm,amsfonts,graphicx,natbib}
\usepackage{epsfig}
\usepackage{bm}
\usepackage{relsize}
\newcommand*\patchAmsMathEnvironmentForLineno[1]{%
      \expandafter\let\csname old#1\expandafter\endcsname\csname #1\endcsname
      \expandafter\let\csname oldend#1\expandafter\endcsname\csname end#1\endcsname
      \renewenvironment{#1}%
         {\linenomath\csname old#1\endcsname}%
         {\csname oldend#1\endcsname\endlinenomath}}%
    \newcommand*\patchBothAmsMathEnvironmentsForLineno[1]{%
      \patchAmsMathEnvironmentForLineno{#1}%
      \patchAmsMathEnvironmentForLineno{#1*}}%
    \AtBeginDocument{%
    \patchBothAmsMathEnvironmentsForLineno{equation}%
    \patchBothAmsMathEnvironmentsForLineno{align}%
    \patchBothAmsMathEnvironmentsForLineno{flalign}%
    \patchBothAmsMathEnvironmentsForLineno{alignat}%
    \patchBothAmsMathEnvironmentsForLineno{gather}%
    \patchBothAmsMathEnvironmentsForLineno{multline}%
    }
\usepackage[mathlines,displaymath]{lineno}
\include{def}
\def\dispmuskip{\thinmuskip= 3mu plus 0mu minus 2mu \medmuskip=  4mu plus 2mu minus 2mu \thickmuskip=5mu plus 5mu minus 2mu}
\def\textmuskip{\thinmuskip= 0mu                    \medmuskip=  1mu plus 1mu minus 1mu \thickmuskip=2mu plus 3mu minus 1mu}
\def\beq{\dispmuskip\begin{equation}}    \def\eeq{\end{equation}\textmuskip}
\def\beqn{\dispmuskip\begin{displaymath}}\def\eeqn{\end{displaymath}\textmuskip}
\def\bea{\dispmuskip\begin{eqnarray}}    \def\eea{\end{eqnarray}\textmuskip}
\def\bean{\dispmuskip\begin{eqnarray*}}  \def\eean{\end{eqnarray*}\textmuskip}

\def\paradot#1{\vspace{1.3ex plus 0.7ex minus 0.5ex}\noindent{\bf\boldmath{#1.}}}

\newtheorem{theorem}{Theorem}

\newtheorem{algorithm}{Algorithm}

\newtheorem{proposition}{Proposition}
\newtheorem{remark}{Remark}

\newcommand{\eps}{\epsilon}
\newcommand{\wh}{\widehat}
\newcommand{\wt}{\widetilde}

\def\E{{\mathbb E}}                         % Expectation
\def\V{{\mathbb V}}
\def\P{{\rm P}}                         % Probability

\def\s{\sigma}

\def\t{\theta}

\def\l{\lambda}

\def\N{{\cal N}}

\def\obs{\text{\rm obs}}

\def\ABC{\text{\rm ABC}}
\def\tr{\text{\rm tr}}

\def\IS{\text{\rm IS}}

\def\Sup{\text{\rm Sup}}

\def\Var{\mathbb{V}}
\def\EABCIS{\text{\rm EABC-$\IS^2$ }}
\newcommand{\dd}{\mathrm{d}}

\begin{document}

\title{Exact ABC using Importance Sampling}
\author{\normalsize Minh-Ngoc Tran and Robert Kohn
\footnote{
Minh-Ngoc Tran is with the Business Analytics discipline, University of Sydney Business School, Sydney 2006 Australia
(minh-ngoc.tran@sydney.edu.au). Robert Kohn is with the UNSW Business School,
University of New South Wales, Sydney 2052 Australia (r.kohn@unsw.edu.au).}
}

\date{\today}
\maketitle
\begin{abstract}
Approximate Bayesian Computation (ABC) is a powerful method for carrying out Bayesian inference when the likelihood is computationally intractable. However, a drawback of ABC is that it is an approximate method that induces a systematic error because it is necessary to set a tolerance level to make the
computation tractable. The issue of how to optimally set this tolerance level has been the subject of extensive research. This paper proposes an ABC algorithm based on importance sampling that estimates expectations with respect to the {\em exact } posterior distribution given the observed summary statistics. This overcomes the need to select the tolerance level. By {\em exact} we mean that there is no systematic error and the Monte Carlo error can be made arbitrarily small by increasing the number of importance samples. We provide a formal justification for the method and study its convergence properties. The method is illustrated in two applications and the empirical results suggest that the proposed ABC based estimators consistently converge to the true values as the number of importance samples increases. Our proposed approach
can be applied more generally to any importance sampling problem where an unbiased estimate of the likelihood is required.

\paradot{Keywords} Approximate Bayesian Computation, Debiasing, Ising model, Marginal likelihood Estimate , Unbiased likelihood Estimate
\end{abstract}

%----------------------------%
\section{Introduction}
%----------------------------%
Many  Bayesian inference problems, including the  calculation of posterior moments and probabilities,
require evaluating an integral of the form
\begin{align}\label{eq:integral of interest}
\E(\varphi|y_\obs) = \int_\Theta \varphi(\t)p(\t|y_\obs)\dd\theta.
\end{align}
In \eqref{eq:integral of interest},  $p(\theta|y_\obs)\propto p(\theta)p(y_\obs|\theta)$
is the posterior distribution of $\theta$,
$y_\obs$ is the observed data, $\theta\in\Theta$ is the vector of model parameters,
and  $\varphi(\theta)$ is a function mapping  $\Theta$ to the real line.
In many problems the likelihood $p(y|\theta) $ is intractable,
either because it cannot be computed or because it is too expensive to compute.
The Approximate Bayesian Computation (ABC) approach was proposed
to overcome this problem as it only requires that we are able
to sample from the model density $y \sim p(\cdot|\theta) $ without being able to evaluate it
\citep[see][]{Tavaré01021997,beaumont2002approximate,marjoram2003markov,Sisson:2011}.

ABC approximates the intractable likelihood $p(y_\obs|\theta)$ by
\beq\label{eq:ABC likelihood 1}
p_{\ABC,\eps}(y_\obs|\t):=\int K_\eps(y-y_\obs)p(y|\t)\dd y
\eeq
with $K_\epsilon(u)$ a scaled kernel density with bandwidth $\epsilon>0$.
If  the original dataset $y$ has a complex structure and is high dimensional,
it is  computationally more efficient and convenient to work with a lower-dimensional
summary statistic $s=S(y)\in \mathbb{R}^d$.
That is, instead of \eqref{eq:ABC likelihood 1}, we work with
\beq\label{eq:ABC likelihood 2}
p_{\ABC,\eps}(s_\obs|\t):=\int K_\eps(s-s_\obs)p(s|\theta)\dd s,
\eeq
where $p(s|\theta)$ denotes the density of summary statistic $s$ and $s_\obs=S(y_\obs)$.
Here, $K_\eps(u)=K(u/\eps)/\eps^d$ with $K(\cdot)$ a $d$-variate kernel density such as a Gaussian density.

There are two major approximations used in ABC.
The first is using  a summary statistic instead of the original data
\begin{align}\label{eq: approximation 1}
p(\theta|y_\obs)& \approx p(\theta|s_\obs)\propto p(\theta)p(s_\obs|\theta),
\end{align}
which is exact if the summary statistic $S(\cdot)$ is sufficient.
The second results from approximating the intractable likelihood $p(s_\obs|\t)$ by $p_{\ABC,\eps}(s_\obs|\t)$,
\beq\label{eq: approximation 2}
p(s_\obs|\t)\approx p_{\ABC,\eps}(s_\obs|\theta)=\int K_\eps(s-s_\obs)p(s|\theta)\dd s.
\eeq
This approximation is exact if $\eps=0$, but setting $\eps=0$ is impractical as
the event that $s=s_\obs$ occurs with probability zero for all but the simplest applications.
Selecting $\eps$ is still an open question because it is usually necessary to trade off between computational load
and accuracy when selecting $\eps$.

All current  ABC algorithms suffer from approximation errors due to approximation \eqref{eq: approximation 1},
if $S(\cdot)$ is not sufficient, and approximation \eqref{eq: approximation 2} if $\epsilon>0$.
Our article proposes an ABC algorithm to estimate \eqref{eq:integral of interest}
that completely removes the error
due to approximation \eqref{eq: approximation 2},
i.e. we are able to estimate expectations with respect to the {\it exact} posterior $p(\theta|s_\obs)$ based on the
summary statistic.
In addition, if $S(\cdot)$ is sufficient, then the
estimated expectations are with respect to the exact posterior $p(\theta|y_\obs)$.

The basic idea is to obtain an unbiased estimator of the likelihood,
based on the debiasing approach of \cite{McLeish:2012} and \cite{Rhee:2013}.
We then construct an importance sampling estimator of the integral \eqref{eq:integral of interest}
and establish its convergence properties.
The unbiasedness allows the importance sampling estimator to
converge almost surely to the true value \eqref{eq:integral of interest} without
suffering from the systematic error associated with the use of $\eps>0$.
We illustrate the proposed method by a Gaussian example and an Ising model example.

We note that our approach can be applied more generally to importance sampling problems where an unbiased estimate of the likelihood is required.

%----------------------------%
\section{Constructing an unbiased estimator using a debiasing approach}
%----------------------------%
Let $\lambda$ be an unknown constant that we want to estimate
and let $\zeta_k,\ k=0,1,...$  be a sequence of biased estimators of $\lambda$, such that it is possible
to generate $\zeta_k$ for each $k$.
We are interested in constructing an unbiased estimator $\wh\l$ of $\lambda$, i.e. $\E(\wh\l)=\lambda$, based on the $\zeta_k$'s, so
 that $\wh\l$ has a finite variance.
We now present  the debiasing approach, proposed independently by \cite{McLeish:2012} and \cite{Rhee:2013},
for constructing such a $\wh\l$.
The basic idea is to introduce randomization into the sequence $\{\zeta_k,k=0,1,2,...\}$ to eliminate the bias.

\begin{proposition}[Theorem 1 of \cite{Rhee:2013}]\label{the: theorem 1}
Suppose that $T$ is a non-negative integer-valued random variable such that $P(T\geq k)>0$ for any $k=0,1,2,...$,
and that $T$ is independent of the $\zeta_k$'s. Let $\varpi_k:=1/P(T\geq k)$.
If
\beq\label{eq:condition 1}
 \sum_{k=1}^\infty \varpi_k{\E\left ((\zeta_{k-1}-\l)^2\right )}<\infty,
\eeq
then
\beqn
\wh \lambda : = \zeta_0+\sum_{k=1}^T \varpi_k(\zeta_{k}-\zeta_{k-1}),
\eeqn
is an unbiased estimator of $\l$
and has the finite variance
\beq\label{eq:variance}
\V(\wh \lambda )=\sum_{k=1}^\infty \varpi_k\Big(\E((\zeta_{k-1}-\l)^2)-\E((\zeta_{k}-\l)^2)\Big)-\E((\zeta_{0}-\l)^2)<\infty.
\eeq
\end{proposition}

%----------------------------%
\section{Exact ABC}
%----------------------------%

%----------------------------%
\subsection{Constructing an unbiased estimator of the likelihood}\label{subsec:constructing}
%----------------------------%

Let $\epsilon_k,\ k=0,1,...$ be a sequence of monotonically decreasing positive numbers
and $n_k$ a sequence of monotonically increasing positive integers  such that  $\epsilon_k\to0$ and $n_k \to \infty$ as  $k\to\infty$.
We estimate the ABC likelihood $p_{\text{ABC},\epsilon_k}(s_\obs|\t)$ based on the $n_k$ pseudo-datasets $s_i^k\sim p(\cdot|\t),\ i = 1,...,n_k, $ as
\begin{align}\label{eq:zeta_k}
\zeta_k& :=\wh p_{\text{ABC},\epsilon_k}(s_\obs|\t)=\frac{1}{n_k}\sum_{i=1}^{n_k}K_{\epsilon_k}(s_i^k-s_\obs).
\end{align}
It is important to note that the pseudo datasets $s_i^k$, $i=1,...,n_k$, can be re-used to compute $\zeta_j$ with $j>k$ as
it is unnecessary that the $\zeta_k$'s in Proposition~\ref{the: theorem 1} are  independent.
This significantly reduces the computational cost when it is expensive to generate these pseudo-datasets from $s\sim p(\cdot|\t)$.

%----------------------------%
\begin{theorem}\label{the: theorem 2}
%----------------------------%
Let $K(\cdot)$ be a $d$-multivariate kernel density, i.e.
$K(x)\geq 0,\;\;\int K(x)dx=1$. We assume that
\begin{align} \label{eq: K properties}
\int xK(x)\dd x & =0,\;\;\sigma^2_K:=\int x'xK(x)\dd x < \infty ,\;\; \sigma^2_R := \int K^2(x)\dd x < \infty ,\;\; \int x'xK^2(x)\dd x < \infty .
\end{align}
Let $T$ be a non-negative integer-valued random variable that is independent of the $\zeta_k, k \geq 0 ,$ and
such that $P(T\geq k)>0$ for any $k \geq 0 $. Let $\varpi_k=1/P(T\geq k)$.
Suppose that $p(s|\theta)$ is twice differentiable in $s$ for every $\t$,
and that
\begin{align}\label{eq:condition 2}
\sum_{k=1}^\infty \varpi_k\left (\epsilon_{k-1}^4+\frac{1}{n_{k-1}\epsilon_{k-1}^d}\right )& <\infty.
\end{align}
Then,
\beqn
{\wh p}(s_\obs|\theta) := \zeta_0+\sum_{k=1}^T \varpi_k(\zeta_{k}-\zeta_{k-1}),
\eeqn
is an unbiased estimator of $p(s_\obs|\theta)$
and has a finite variance.
%----------------------------%
\end{theorem}
%----------------------------%

We now use Theorem~\ref{the: theorem 2} to
construct an unbiased estimator $\wh p(s_\obs|\theta)$ of the posterior $p(s_\obs|\theta)$
by designing $T$, $\eps_k$ and $n_k$ to satisfy the conditions of Theorem~\ref{the: theorem 2}.
Let $T$ be a non-negative integer-valued random variable such that $\P(T=k):=\rho(1-\rho)^k$, $k=0,1,...$ for $0<\rho<1$.
This choice means that the closer $\rho$ is to $0$, the bigger the values that $T$ is likely to take.
Then $\varpi_k=1/P(T\geq k)=1/(1-\rho)^{k}$.
Let $\tau$ be a number such that $0<\tau<1$.
If we select
\beqn
\eps_k:=[\tau(1-\rho)]^{\frac{k+1}{4}} \quad \text{and} \quad \;n_k:=\left\lceil\frac{1}{[\tau(1-\rho)]^{(k+1)(1+d/4)}}\right\rceil,
\eeqn
then
\beqn
\sum_{k=1}^\infty \varpi_k\left (\epsilon_{k-1}^4+\frac{1}{n_{k-1}\epsilon_{k-1}^d}\right )<2\sum_{k=1}^\infty\tau^k<\infty.
\eeqn
That is, condition \eqref{eq:condition 2} is satisfied.

From \eqref{eq:variance} and \eqref{eq:mean squared error}, after some algebra, the variance $\V({\wh p}(s_\obs|\theta))$ is approximately
\beqn
\V({\wh p}(s_\obs|\theta))\approx(C_1+C_2)\big(1-\tau(1-\rho)\big)\frac{\tau}{1-\tau}-(C_1+C_2)\big(\tau(1-\rho)\big)^{1/4},
\eeqn
with $C_1$ and $C_2$ positive constants in the proof of Theorem \ref{the: theorem 2}.
The first term, which dominates the second term, is a monotonic increasing function of $\tau$ and $\rho$.
So the variance $\V({\wh p}(s_\obs|\theta))$ will be small if $\rho$ and $\tau$ are close to 0.
However, small $\rho$ and $\tau$ lead to a large $k$ and hence a large $n_k$, especially if  $d $ is large.
We can reduce the variance of the unbiased estimator ${\wh p}(s_\obs|\theta)$  by using $\overline {{\wh p}(s_\obs|\theta)}=({\wh p}(s_\obs|\theta)_1+\cdots + {\wh p}(s_\obs|\theta)_{n_\text{rep}})/n_\text{rep}$,
with the ${\wh p}(s_\obs|\theta)_i$ independent replications of ${\wh p}(s_\obs|\theta)$. Then $\E\big(\overline {{\wh p}(s_\obs|\theta)}\big)=p(s_\obs|\theta)$. This
approach to estimating $p(s_\obs|\theta) $ has the important advantage that it automatically gives an estimate of
$\V({\wh p}(s_\obs|\theta)) $ and hence $\V(\overline {{\wh p}(s_\obs|\theta)})=\V({\wh p}(s_\obs|\theta) /n_\text{rep}$, i.e.,
\begin{align*}
{\wh \V}({\wh p}(s_\obs|\theta))  & = \frac{\mathlarger \sum_{i=1}^{n_\text{rep}}\left ({\wh p}(s_\obs|\theta)_i - \overline {{\wh p}(s_\obs|\theta)})  \right )^2}{(n_\text{rep}-1)}
\quad \text{ and} \quad {\wh \V} (\overline {{\wh p}(s_\obs|\theta)}))=\frac{{\wh \V} ({\wh p}(s_\obs|\theta)}{n_\text{rep}}.
\end{align*}

%----------------------------%
\subsection{Exact ABC with $\IS^2$}
%----------------------------%
Define $\pi(\theta):=p(\t|s_\obs)$ and
let $\wh p(s_\obs|\t,u)$ be the unbiased estimator of $p(s_\obs|\t)$ obtained using the debiasing approach described in the previous section, and $u \in \mathcal{U}$ is the set of uniform random variables used to generate $T $ and $\zeta_0, \dots, \zeta_T$.
We denote by $p(u|\t,s_\obs)$ the density of $u$ and sometimes write $p(u|\t,s_\obs)$ as $p(u|\t)$ for notational simplicity.
If the unbiased estimator $\wh p(s_\obs|\theta,u)$ is non-negative almost surely for each $\theta$, then we could use the pseudo-marginal Metropolish-Hastings (PMMH) algorithm \citep{Andrieu:2009}
to sample  from the posterior $p(\theta|s_\obs)$.
In general, however, the debiased estimator $\wh p(s_\obs|\theta,u)$ can be negative,
so it is mathematically invalid to use PMMH in our situation.
See \cite{Jacob:2015} for a detailed discussion.

Suppose that we wish to estimate the expectation of the function $\varphi(\theta)$ on $\Theta$ with respect to the posterior distribution, i.e.,
\begin{align*}
\E_\pi(\varphi) & =\int_\Theta \varphi(\theta)\pi(\theta)\dd\theta
= \frac{ \int_\Theta  \varphi(\theta)p(s_\obs|\theta)p(\theta) \dd \theta} { \int_\Theta p(s_\obs|\theta)p(\theta) \dd \theta}.
\end{align*}
 Then,
\begin{align*}
\E_\pi(\varphi) & = \frac{ \int_\Theta\int_\mathcal{U}  \varphi(\theta){\wh p}(s_\obs|\theta, u) p(\theta)p(u|\theta,s_\obs) \dd \theta\dd u }
{ \int_\Theta\int_\mathcal{U} {\wh p}(s_\obs|\theta, u)p(\theta)p(u|\theta,s_\obs) \dd \theta\dd u }.
\end{align*}
Let $g_\IS(\t)$ be an importance density on $\Theta$.
For a function $h(\theta)$ of $\theta \in \Theta$, define
\begin{align*}
I(h) & := \int_\Theta  h(\theta)  p(s_\obs|\theta)p(\theta) \dd \theta
 = \int_\Theta \int_\mathcal{U}  h(\theta) \frac{{\wh p}(s_\obs|\theta, u)p(\theta)}{g_{\IS}(\theta)} g_{\IS}(\theta) p(u|\theta,s_\obs) \dd \theta \dd u
\end{align*}
which is unbiasedly estimated by
\begin{align}
{\wh I}(h)  & := \frac1M \sum_{i=1}^M h(\theta_i) {\wh w} (\theta_i,u_i), \notag \\
\intertext{where}
\theta_i \sim g_{\IS}(\cdot),\ u_i & \sim p(\cdot |\theta_i, s_\obs) \quad \text{and} \quad {\wh w} (\theta_i,u_i) := \frac{{\wh p}(s_\obs|\theta,u_i) p(\theta_i)}{g_{\IS}(\theta_i) }. \label{eq:weight}
\end{align}
We now define the estimate of $\E_\pi(\varphi)$ as
\begin{align}
\widehat{ \E_\pi(\varphi)}& := \frac{{\wh I}(\varphi) }{{\wh I}(1)  }  \label{eq:EABC-IS2}.
\end{align}
In this form, the estimator $\widehat{ \E_\pi(\varphi)}$ is similar to the $\IS^2$ estimator introduced in \cite{Tran:2013},
who propose an importance sampling procedure when the likelihood is intractable but a non-negative unbiased estimator of the likelihood is available.

We now summarize the algorithm for estimating $\E_\pi(\varphi)$,
and refer to it as the Exact ABC algorithm based on an $\IS^2$ approach, or \EABCIS for short.
\begin{algorithm}[\EABCIS algorithm] For $i=1,...,M$
\begin{itemize}
\item Generate $\theta_i\sim g_\IS(\cdot)$, $u_i\sim p(\cdot|\theta_i,s_\obs)$  and compute $\wh p(s_\obs|\t_i,u_i)$.
\item Compute the weights $\wh w(\theta_i,u_i)$ as in \eqref{eq:weight}.
\end{itemize}
Compute the \EABCIS estimator $\wh{\E_\pi(\varphi)}$ of $\E_\pi(\varphi)$ as in \eqref{eq:EABC-IS2}.
\end{algorithm}

\begin{remark}\label{rem: remark many}
As with all importance sampling, it is straightforward to estimate several expectations simultaneously at almost the same cost as one expectation, because the weights
$\wh w(\t_i,u_i)$ are the same.
\end{remark}

To obtain a strong law of large numbers and a central limit theorem for $\widehat{ \E_\pi(\varphi)}$ we define $\xi (\theta,u) := {\wh p}(s_\obs|\theta,u)/p(s_\obs|\theta)$, so that $\E_{u\sim p(\cdot|\theta)}( \xi (\theta,u))=1$.

%-------------------------------------------%
\begin{theorem} \label{th: my thm}
%-------------------------------------------%
Suppose that $\mathrm{Sup}(\pi) \subseteq \mathrm{Sup}(g_{\IS})$, where $\mathrm{Sup}$ means support.
\begin{enumerate}
\item [(i)]
If $\E_\pi(|\varphi(\theta)|)  < \infty$, then $\widehat{ \E_\pi(\varphi)} \to \E_\pi(\varphi)$ almost surely as $M \to \infty$.
\item [(ii)]
If $\E_{g_{\IS}} \left ( \frac{\E_{u\sim p(\cdot|\theta)}(\xi^2(\theta,u))  \varphi(\theta)^2\pi(\theta)^2}{g_{\IS}^2(\theta)} \right )  < \infty$
then $\sqrt M \left ( \widehat{ \E_\pi(\varphi)} - \E_\pi(\varphi)\right ) \to \N(0, \sigma_\varphi^2)$
as $M \to \infty$, where
\begin{align}  \label{eq: clt}
\sigma_\varphi^2:= \E_{ g_{\IS}}\left ( \frac{ \pi^2(\theta)}{g_{\IS}^2(\theta)} (\varphi(\theta)-\E_\pi(\varphi))^2 \E_{u\sim p(\cdot|\theta)}(\xi^2(\theta,u)) \right ).
\end{align}
If  we can evaluate $p(s_\obs|\theta)$ so that $\xi = 1$, then $\sigma_\varphi^2 =\E_{ g_{\IS}}\left ( \frac{ \pi^2(\theta)}{g_{\IS}^2(\theta)} (\varphi(\theta)-\E_\pi(\varphi))^2 \right )$ is the variance of the noiseless importance sampler.
\item [(iii)] ${\wh {\sigma_\varphi^2}}$ is a consistent estimator  of $\sigma_\varphi^2$, where
\begin{align}
{\wh {\sigma_\varphi^2}} & := \frac{1}{M{\wh p}(s_\obs)^2}
 \sum_{i=1}^M \big(\varphi(\theta_i)- \wh{\E_\pi (\varphi)}\big)^2 {\wh w}^2(\theta_i,u_i),  \notag
\intertext{and}
{\wh p}(s_\obs)&  := \frac1M \sum_{i=1}^M {\wh w} (\theta_i,u_i).  \label{eq: marg likelihood}
\end{align}
\end{enumerate}
\end{theorem}

\begin{remark}
We note that ${\wh p}(s_\obs)$ in \eqref{eq: marg likelihood}
is an estimate of the marginal likelihood $p(s_\obs)$, which can be used for model comparison. It is straightforward to obtain
this marginal likelihood estimate and an estimate of its standard error and
we can readily show that ${\wh p}(s_\obs)$
converges to $p(s_\obs)$ as $M \rightarrow \infty$. It is usually difficult to accurately
estimate the marginal likelihood and its standard error using competing ABC approaches.
\end{remark}
%----------------------------%
\section{Examples}
%----------------------------%

%----------------------------%
\subsection{A Gaussian example}
%----------------------------%
This  example is discussed by \cite{Sisson:2011} who consider a univariate Gaussian model $y\sim \N(\theta,1)$, with
 $y_\obs=0$ and  a diffuse prior $p(\theta)\propto 1$. Here, the posterior is $\pi(\theta)=p(\theta|y_\obs)=\N(0,1)$ and
 the summary statistics $s=S(y)=y$ is sufficient.
We are interested in estimating the posterior  noncentral second moment of $\theta$,
\begin{align*}
\E(\theta^2|y_\obs) = \int \theta^2p(\theta|y_\obs)\dd\theta =1.
\end{align*}
We take the kernel $K(\cdot)$  as the standard normal density, so
the ABC likelihood $p_{\ABC,\eps}(y_\obs|\t)$ in \eqref{eq:ABC likelihood 2} can be computed analytically,
and the ABC posterior is $p_{\ABC,\eps}(\theta|y_\obs)\propto p(\t)p_{\ABC,\eps}(y_\obs|\t)=\N(0,1+\eps^2)$.
So setting aside the Monte Carlo error, standard ABC procedures estimate $\E(\theta^2|y_\obs)$ by $1+\eps^2$,
which always suffers from a systematic error whenever $\eps>0$.

\begin{table}[h]
\centering
\vskip2mm
{\small
\begin{tabular}{c|cccc}
\hline
$M$  		&1000		&10,000		&100,000 & 1,000,000\\
\hline
\EABCIS estimate&1.0065 (0.0733)&1.0044 (.0245)	&1.0008 (0.0111)&1.0000 (0.0002)\\
\hline
\end{tabular}
}
\caption{\EABCIS estimates of $\E(\t^2|y_\obs)=1$ for various numbers of samples $M$. The numbers in brackets are standard errors} \label{tab:toy example}
\end{table}
To run the \EABCIS algorithm, we
estimate the likelihood $p(y_\obs|\t)$ unbiasedly
using the debiasing approach in Section \ref{subsec:constructing}
with $\rho=0.4$, $\tau=0.2$
and the importance density $g_\IS(\t)=\N(0,2)$.
The number of replications $n_\text{rep}$ is selected such that the variance
$\V(\log|\wh p(y_\obs|\bar\t)|)\approx 1$ with $\bar\t=0.5$.
This is motivated by the $\IS^2$ theory in \cite{Tran:2013}
who show that the optimal variance of the log-likelihood estimators is 1
in order to minimize the overall computational cost.

Table \ref{tab:toy example} shows the \EABCIS estimates of $\E(\theta^2|y_\obs)$ for various numbers of samples $M$.
The results suggest empirically that the estimates consistently get closer to the true value
as $M$ increases.
This attractive property of the \EABCIS is contrasted with other ABC algorithms
where a systematic error always exists no matter how large $M$ is.

%----------------------------%
\subsection{Ising model}
%----------------------------%
Our second example is the Ising model on a rectangular lattice of size $L\times W$
with data $y_{i,j}\in\{-1,1\}$ and likelihood
\beqn
p(y|\theta )=\frac{\exp(\theta S(y))}{C(\t)},
\eeqn
where $S(y)=\sum_{i=1}^{L-1}\sum_{j=1}^W y_{i,j}y_{i+1,j}+\sum_{i=1}^{L}\sum_{j=1}^{W-1} y_{i,j}y_{i,j+1}$;
see \cite{Moller:2006}.
The likelihood $p(y|\theta)$ has $S(y)$ as sufficient statistic and is considered intractable as computing the normalising constant $C(\theta)$ for each $\theta$
is infeasible
for large lattices.
However, one can generate data $y$ from the Ising model $y\sim p(\cdot|\theta)$
using, for example, perfect simulation or Monte Carlo simulation.
We note that $S(y)$ is a sufficient statistic for $\t$.

In this example, we set $L=W=50$ and generate a data set $y_\obs$ using $\theta=0.5$.
Our task is to estimate the posterior mean of $\t$, given $y_\obs$.
As in \cite{Moller:2006}, we use a uniform prior $U(0,1)$ for $\t$.
For this Ising model, an exact MCMC is available
for sampling from the posterior $p(\t|y_\obs)$ \citep{Moller:2006},
which we use as the \lq\lq gold standard\rq\rq{} for comparison.
We run this exact MCMC algorithm for 1,000,000 iterations and obtain an estimate of 0.5099.
for the posterior mean $\E(\theta|y_\obs)=\int \t p(\theta|y_\obs)\dd\theta$.
The number in brackets is the standard deviation.

\iffalse
The standard deviation is estimated as follows.
Let $\{\theta_i,i=1,...,M\}$ be the generated iterates from the Markov chain.
Define the integrated autocorrelation time as
\beqn
\text{IACT} = 1+2\sum_{t=1}^\infty\rho_t,
\eeqn
where $\rho_t=\text{corr}(\t_1,\t_{t+1})$ is the autocorrelation of the chain at lag $t$.
The posterior mean $\E(\t|y_\obs)$ is estimated by $\sum_i \t_i/M$ whose variance is
\beqn
\frac{\sigma^2}{M}\left(1+2\sum_{t=1}^{M-1}\Big(1-\frac tM\Big)\rho_t\right)\approx\frac{\sigma^2}{M}\Big(1+2\sum_{t=1}^\infty\rho_t\Big)=\text{IACT}\cdot\frac{\sigma^2}{M},
\eeqn
where $\sigma^2$ is the variance $\V(\t|y_\obs)$.
We estimate the IACT by the sample autocorrelations as in \cite{Tran:2014}
and estimate $\s^2$ by \EABCIS.
\fi

The \EABCIS estimate of $\E(\t|y_\obs)$, based on $M=200,000$ samples of $\t$, is 0.5099 (0.0001)
which is equal to (up to 4 decimal places)  the estimate given by the exact MCMC algorithm.

We now use PMMH to sample from the ABC posterior $p_{\ABC,\eps}(\t|y_\obs)$, for various $\eps=10,\ 1$ and 0.1,
with the ABC likelihood $p_{\ABC,\eps}(y_\obs|\t)$ in \eqref{eq:ABC likelihood 1}
estimated unbiasedly by
\beqn
\wh p_{\text{ABC},\epsilon}(s_\obs|\t)=\frac{1}{n}\sum_{i=1}^{n}K_{\epsilon}(s_i-s_\obs),\;\;s_i\sim p(\cdot|\t).
\eeqn
For each $\eps$, the number of pseudo datasets $n$ is tailored such that the acceptance rate is about 0.23.
The ABC-PMMH estimates of the posterior mean $\E_{\theta\sim p_{\ABC,\eps}(\t|y_\obs)  }(\t|y_\obs)$, based on 100,000 iterations, are 0.5094 (0.0002), 0.5108 (0.0003) and 0.5100 (0.0003) respectively.
These estimates get closer to the \lq \lq gold standar\lq \lq  estimate 0.5099 when $\eps$ decreases.
Note that the smaller the value of $\eps$, the greater the computational cost
as we need a bigger $n$ in order for the Markov chain to mix well.

%----------------------------%
\section{Discussion}
%----------------------------%
Our article presents the \EABCIS approach for estimating expectations with respect to the exact posterior distribution conditional on the observed summary statistic.
The \EABCIS estimators do not suffer from a systematic error inherent in standard ABC algorithms due to the use of tolerance $\eps>0$.
Our approach generalises directly to other applications of importance sampling where the likelihood is intractable but
an unbiased estimator of the likelihood can be used.
\section*{Appendix: Proofs}
\begin{proof}[Proof of Theorem~\ref{the: theorem 2}] For a fixed $\theta$, let $\lambda=p(s_\obs|\t)$.
We first show that
\begin{align} \label{eq: prel thm 2}
\Big(p_{\text{ABC},\eps_k}(y|\t)-\l\Big)^2=\frac14\eps_k^4\s_K^4\Big(\tr(\nabla^2p(s_\obs|\t))\Big)^2+o(\eps_k^4).
\end{align}
\begin{align*}
p_{\text{ABC},\eps_k}(s_\obs|\t)&=\frac{1}{\eps_k^d}\int K(\frac{s-s_\obs}{\eps_k})p(s|\t)\dd s\\
&=\int K(w)p(s_\obs+\eps_k w|\t)\dd w,\;\;\mathrm{where} \quad w:=\frac{s-s_\obs}{\eps_k}\\
&=\int K(w)\Big(p(s_\obs|\t)+\eps_k w'\nabla p(s_\obs|\t)+\frac12\eps_k^2w'\nabla^2p(s_\obs|\t)w+o(\eps_k^2)\Big)\dd w\\
&=p(s_\obs|\t)+\frac12\eps_k^2\s_K^2\tr(\nabla^2p(s_\obs|\t))+o(\eps_k^2),
\end{align*}
which gives \eqref{eq: prel thm 2}.
Similarly,
\beqn
\V(\zeta_k)=n_k^{-1}\eps_k^{-d}R_Kp(s_\obs|\t)+o(n_k^{-1}\eps_k^{-d}),
\eeqn
where $R_K=\int K(x)^2 \dd x$.

Then,
\bea\label{eq:mean squared error}
\E\left ((\zeta_k-\l)^2\right )&=&\V(\zeta_k)+\big(p_{\text{ABC},\epsilon_k}(y|\t)-\l\big)^2\notag\\
&=&C_1\eps_k^4+C_2n_k^{-1}\eps_k^{-d}+o(\eps_k^4+n_k^{-1}\eps_k^{-d}),
\eea
and \eqref{eq:condition 2} implies \eqref{eq:condition 1}.
The proof then follows from Proposition~\ref{the: theorem 1}
\end{proof}
\begin{proof} [Proof of Theorem~\ref{th: my thm} ]
The proof is similar to that of Theorem 1 in \cite{Tran:2013}.
Let ${\wt g}_{\IS}(\theta,u): = g_{\IS}(\theta)p(u|\theta,s_\obs)$ and ${\wt \pi} (\theta, u) := \pi(\theta) p(u|\theta,s_\obs)$.
The condition $\Sup(\pi)\subseteq\Sup(g_\text{IS})$ implies that $\Sup(\wt\pi)\subseteq\Sup(\wt g_\text{IS})$.
This, together with the existence and finiteness of $\E_\pi(\varphi)$ ensure that
\beqn
\E_{\wt g_\IS}[\varphi(\theta_i){\wh w}(\theta_i,u_i)]=p(s_\obs)\E_\pi(\varphi)\;\;\text{and}
\;\; \E_{\wt g_\IS}[\wh w(\theta_i,u_i)]=p(s_\obs)
\eeqn
exist and are finite.
Result (i) then follows immediately from \eqref{eq:EABC-IS2} and the strong law of large numbers.

To prove (ii), write
\begin{align*}
\wh{\E_{\pi}(\varphi)}-\E_\pi(\varphi) & = \frac{\frac1M \sum_{i=1}^M \big(\varphi(\t_i)-
\E_\pi(\varphi)\big) \wh w(\theta_i,u_i)} {\frac1{M}\sum_{i=1}^{M}\wh w(\theta_i, u_i) }
 = S_M /{\wh p(s_\obs)}, \\
\text{where} \; \;  S_M & = M^{-1} \sum_{i=1}^M X(\theta_i,u_i) ,   \quad \text{with} \quad
 X(\theta, u) =  \big(\varphi(\t_i)-
\E_\pi(\varphi)\big) \wh w(\theta_i,u_i)
\end{align*}
The $X_i:= X(\theta_i,u_i)$ are independently and identically distributed and it is straightforward to check that  $\E_{\wt g_\IS}(X)=0$.
\begin{align*}
\Var_{\wt g_\IS}(X)&= \E _{\wt  g_\IS}(X^2)\\
& = \E_{g_{\IS}} \left ( \E_{u \sim p(\cdot |\theta, s_\obs)} (X^2)\right ) \\
&  = \E _{  g_{\IS}}\left ( \Big ( \big(\varphi(\t_i)-
\E_\pi(\varphi)\big) \frac{p(\theta)p(s_\obs|\theta)}{g_{\IS}(\theta)}\Big)^2 \E_{u \sim p(\cdot |\theta, s_\obs)} (\xi^2)\right ) \\
& = p(s_\obs)^2 \E _{  g_{\IS}}\left ( \Big ( \big(\varphi(\t_i)-
\E_\pi(\varphi)\big) \frac{\pi(\theta)}{g_{\IS}(\theta)}\Big)^2 \E_{u \sim p(\cdot |\theta, s_\obs)} (\xi^2)\right ) \\
& = p(s_\obs)^2 \sigma^2_\varphi
\end{align*}
By the central limit theorem for a sum of independently and identically distributed random variables
with a finite second moment, $\sqrt{M}S_{M}\stackrel{d}{\to}\N(0, p(s_\obs)^2 \sigma^2_\varphi )$.
By (i) and Slutsky's theorem,
\beqn
\sqrt{M}\Big(\wh{\E_\pi(\varphi)}-\E_\pi(\varphi)\Big) =\frac{\sqrt{M}S_{M}}{ {\wh p}(s_\obs)} \stackrel{d}{\to}\N(0,\sigma^2_\varphi ) %
\eeqn
To prove (iii), it is sufficient to show that
\begin{align*}
\wh \sigma^2_\varphi & := \frac{1}{M\wh p(s_\obs)} \sum_{i=1}^M \big( \varphi(\theta_i) -\wh{\E_\pi( \varphi)}\big)^2 {\wh w}^2(\theta_i,u_i) \\
& \stackrel{a.s.}{\longrightarrow} \frac{\E_{\wt g_\IS}(X^2)}{p(s_\obs)^2} =\sigma^2_\varphi.
\end{align*}
%The result follows because the $X_i$ are iid with a finite variance so that $\wt \sigma^2_\varphi\stackrel{a.s.}{\to} \E_{\wt g_{\IS}}(X^2)/p^2(s_\obs) $ by the
%strong law of large numbers.
\end{proof}

\bibliographystyle{apalike}
\bibliography{references}

\end{document}